\documentclass[runningheads]{llncs}

\usepackage{amsmath}
\usepackage{amsfonts}
\usepackage{amssymb}
\usepackage{booktabs}
\usepackage{graphicx}
\usepackage[colorlinks]{hyperref}
\usepackage[switch]{lineno}
\usepackage[dvipsnames]{xcolor}

\graphicspath{{figures/}}

\newcommand{\appendixref}[1]{\hyperref[#1]{Appendix~\ref*{#1}}}


\DeclareMathOperator{\E}{\mathbb{E}}     
\DeclareMathOperator*{\argmax}{arg\,max}

\newcommand{\DB}{\textbf{DB}}
\newcommand{\stand}{\mathcal{S}}
\newcommand{\numcands}{k}
\newcommand{\altorder}{alt-order}

\title{Efficient Weighting Schemes for Auditing Instant-Runoff
Voting Elections\thanks{Authors listed alphabetically.  To appear in the 9th
Workshop on Advances in Secure Electronic Voting (Voting'24), on 8 March
2024.}}

\titlerunning{Efficient Weighting Schemes for Auditing IRV Elections}

\author{
Alexander Ek       \inst{1}  \orcidID{0000-0002-8744-4805}  \and
Philip B. Stark    \inst{2}  \orcidID{0000-0002-3771-9604}  \and
Peter J. Stuckey   \inst{3}  \orcidID{0000-0003-2186-0459}  \and
Damjan Vukcevic    \inst{1}  \orcidID{0000-0001-7780-9586}
}

\authorrunning{Ek A, Stark PB, Stuckey PJ, Vukcevic D}

\institute{
Department of Econometrics and Business Statistics, Monash University, Clayton,
Australia
\and
Department of Statistics, University of California, Berkeley, CA, USA
\and
Department of Data Science and AI, Monash University, Clayton, Australia \\
\email{damjan.vukcevic@monash.edu}
}

\begin{document}

\maketitle

\begin{abstract}
Various risk-limiting audit (RLA) methods have been developed for
instant-runoff voting (IRV) elections.  A recent method, AWAIRE, is the
first efficient approach that can take advantage of but does not require cast
vote records (CVRs).  AWAIRE involves adaptively weighted averages of test
statistics, essentially ``learning'' an effective set of hypotheses to test.
However, the initial paper on AWAIRE only examined a few weighting schemes and
parameter settings.

We explore schemes and settings more extensively, to identify and recommend
efficient choices for practice.  We focus on the case where CVRs are not
available, assessing performance using simulations based on real election data.

The most effective schemes are often those that place most or all of the weight
on the apparent ``best'' hypotheses based on already seen data.  Conversely,
the optimal tuning parameters tended to vary based on the election margin.
Nonetheless, we quantify the performance trade-offs for different choices
across varying election margins, aiding in selecting the most desirable
trade-off if a default option is needed.

A limitation of the current AWAIRE implementation is its restriction to a small
number of candidates---up to six in previous implementations.  One path to a
more computationally efficient implementation would be to use lazy evaluation
and avoid considering all possible hypotheses.  Our findings suggest that such
an approach could be done without substantially compromising statistical
performance.
\end{abstract}


\section{Introduction}

Elections are crucial to democracy.  Ensuring that elections truly reflect the
preferences of the population should be a cornerstone of democratic governance.
While there are many forms of elections, \emph{ranked-choice} or
\emph{preferential} voting allows voters to express preferences among some or
all candidates, rather than simply voting for a single candidate.
Instant-runoff voting (IRV) is used in elections in many countries, including
Australia, Ireland, and the USA.

While ranked-choice voting captures more of the preferences of voters, assuring
that their preferences are followed requires ensuring that the reported
\emph{outcome} of an election is correct, that is the reported winner of the
election is the winner if we followed the election process correctly on the
correct set of ballots.  Risk-limiting audits (RLAs) are a way of checking that
a reported election outcome is correct. As opposed to other auditing methods,
RLAs guarantee with some minimum probability that they will correct an
incorrect reported outcome of an election, and never alter a correct outcome.
The \emph{risk limit}, denoted by $\alpha$, is the maximum chance that a wrong
outcome will not be corrected.

The first RLA approach to auditing IRV elections, RAIRE~\cite{blom2019raire},
makes use of a digitised record of the votes in the election (the \emph{cast
vote records} (CVRs)) to generate a set of ``assertions'' that, if true,  imply
that the reported winner really won.  These assertions are currently used in
the SHANGRLA framework for RLAs \cite{shangrla}, and have been used to audit
actual elections \cite{blomEtal20}.  More recently, an alternative approach to
RLAs for IRV elections that does not require CVRs, AWAIRE~\cite{ek2023awaire},
was published.  AWAIRE has the advantage that many IRV elections are tabulated
by hand,\footnote{%
Most, but not all, lower house elections in Australia are hand-counted IRV
contests.}
and no digitised record of the ballots is actually made, so RAIRE is not
applicable in these circumstances.  While RAIRE commits to a set of assertions
to check before the audit starts, AWAIRE adapts to the voter preferences
observed in the audit sample as the audit progresses, identifying a sufficient
set of assertions that are efficient to test statistically.  AWAIRE is also
more resilient than RAIRE when the reported outcome is correct but the
digitised vote records lead to an erroneous elimination order.

AWAIRE uses a \emph{weighting scheme} to adapt the assertions it will
concentrate on as the audit progresses, and more and more observations of
ballots are seen. In the original AWAIRE paper~\cite{ek2023awaire}, the authors
consider a few simple weighting schemes and a single default choice of
parameters used for ALPHA \cite{stark2023alpha}, the statistical test used to
test whether assertions are correct (within a statistical limit on the
acceptable chance of error).  In this paper we:
\begin{itemize}
\item Expand upon AWAIRE  by investigating more weighting schemes and exploring
    how the margin of victory affects which weighting scheme is best.
\item Investigate the effect of ALPHA tuning parameters on audit efficiency.
\end{itemize}


\section{Auditing IRV Contests Using AWAIRE}

\subsection{Instant-runoff voting (IRV)}

In an IRV contest, voters write on their ballot an ordering of (possibly a
subset of) the candidates based on the voter's preference.

The votes are tabulated as follows: Initially, each ballot counts as a single
vote for its first-choice candidate on that ballot.  The candidate with the
fewest first-choice votes is eliminated, while the others remain in the race.
Every ballot that ranked the eliminated candidate first is now instead counted
as a vote for its second choice, i.e., it becomes a vote for the top-ranked
candidate remaining in the race.  This process continues until only one
candidate remains, the winner.  As a ballot need not list every candidate, if
at any point there are only eliminated candidates listed on a ballot, then the
ballot is \emph{exhausted} and no longer contributes any votes.  The above
tabulating process leads to an \emph{elimination order}: the order in which
candidates are eliminated, with the last candidate in the order being the
winner.

In order to audit an IRV election we need to show that it would be unlikely
that any candidate other than the reported winner actually won.

\subsection{The AWAIRE Framework}

AWAIRE is an RLA method for IRV elections that does not require an electronic
record of the votes on each ballot (CVRs) to proceed.
In brief:
\begin{itemize}
\item AWAIRE tracks every elimination order that yields a winner other than the
    reported winner; we refer to these orders as \emph{\altorder(s)}.  If there
    is sufficiently strong evidence (based on a pre-specified risk limit) that
    no \altorder{} is correct, then the audit stops without a full hand count
    and AWAIRE concludes that the reported winner really won.
\item Each \altorder{} is characterised by a set of \emph{requirements}:
    necessary conditions for that elimination order to be correct.  If the data
    refutes at least one requirement for each \altorder, then the reported
    outcome is confirmed.
\item A \emph{test supermartingale} is constructed for each requirement.
    A test supermartingale is also constructed for each \altorder{}, by
    defining each new term as a (predictable) convex combination of the terms
    in the test supermartingales for each requirement in the \altorder{}.
\item As the audit progresses, the convex combination for each \altorder{} is
    updated to give more weight to the test supermartingales that are giving
    the strongest evidence that their corresponding requirements are false.
\item The audit stops when the test supermartingale for every \altorder{}
    exceeds $1 / \alpha$, or when every ballot has been inspected and the
    correct outcome is known.
\item The process described above has risk limit $\alpha$.
\end{itemize}

\subsection{Test Supermartingales}

A supermartingale is a mathematical model of a gambler's fortune in a sequence
of wagers that are fair or biased against the gambler.  Specifically, a
supermartingale is a stochastic process $(M_t)_{t \in \mathbb{N}}$ (e.g.,
fortune after $t$ bets) with respect to another stochastic process $(X_t)_{t
\in \mathbb{N}}$ (e.g., a series of $t$ coin flips that we bet on), where the
conditional expected value of the next observation, given all past
observations, is not greater than the current  observation; that is, $\E(M_t
\mid X_1, \dots, X_{t-1}) \leqslant M_{t-1}$.

A test statistic that is a nonnegative supermartingale starting at 1 when a
hypothesis is true can be used to test that hypothesis.  We call such a process
a \emph{test supermartingale} for the hypothesis.  By Ville's inequality
\cite{ville39}, which generalises Markov's inequality to nonnegative
supermartingales, the chance that a test supermartingale ever exceeds
$1/\alpha$ is at most $\alpha$ if the null hypothesis is true.  Hence,
rejecting the null hypothesis when $M_t \geqslant1/\alpha$ for some time $t$ is
a level $\alpha$ test of the hypothesis.

In words, suppose that a gambler starts with a fortune of \$1 and is not
allowed to go into debt.  The gambler bets on a sequence of games.  The chance
that the gambler's fortune ever gets to \$$1/\alpha$ is at most $\alpha$ if the
games are fair or biased against the gambler.  If the gambler succeeds in
amassing a fortune of, say, \$1,000, then that is quite strong evidence that
some of the games had odds that were favorable to the gambler---that the games
were not all fair or sub-fair.  Had the games all been fair or sub-fair, the
chance of reaching a fortune of \$1,000 would be at most $1/1000 = 0.001$.

\subsection{Hypotheses and Requirements}

The process of auditing an IRV contest can be expressed as a collection of
hypothesis tests.  In the AWAIRE framework, we try to reject each \altorder{}.
For an election with $\numcands$ candidates, there are $m = \numcands! -
(\numcands - 1)!$ \altorder{s}.  Let $H_0^j$ denote the hypothesis that the
$j$th \altorder{} is correct.  Then, to conclude that the reported winner
really won without the audit becoming a full hand count, we need to reject the
composite null hypothesis
\[
    H_0 = H_0^1 \cup \dots \cup H_0^m.
\]
To reject an \altorder{} in the AWAIRE framework, we need to reject one or more
of its \emph{requirements}, relations that must hold for the \altorder{} to
hold (i.e., they are necessary and sufficient for \altorder{} $i$ to be
correct).  Hence, if we can reject one requirement with risk $\alpha$, then we
can reject the \altorder{} with risk $\alpha$.  We must reject
\[
    H_0^i = R_i^1 \cap R_i^2 \cap \dots \cap R_i^{r_i},
\]
where $R_i^1, R_i^2, \dots, R_i^{r_i}$ are the requirements of \altorder{} $i$.

In IRV, each requirement is comprised of so-called \emph{directly beats}
assertions.  The assertion $\DB(i, j, \stand)$, where $\stand \supseteq \{i,
j\}$, holds if candidate~$i$ has more votes than candidate~$j$, given that only
the candidates in $\stand \supseteq \{i, j\}$ have not been eliminated.  If the
assertion is true, then it means that $j$ cannot be the next eliminated
candidate (as $j$ would be eliminated before $i$) if only the candidates
$\stand$ remain standing. For more details about the assertions and how to
build test supermartingales for individual assertions, we refer the reader
to~\cite{ek2023awaire}.\footnote{%
An understanding of these details is not necessary for the current paper.}

At each time $t$, a ballot is drawn without replacement.  Every ballot is
encoded (via an \emph{assorter}, see~\cite{shangrla}) as either evidence
against (value $1$), for (value $0$), or neutral to (value $1/2$) a requirement
being true.  Each requirement can be expressed as the hypothesis that the mean
of a list of encoded ballots is less than $1/2$.  We test the requirement using
the ALPHA test supermartingale~\cite{stark2023alpha}.

ALPHA involves specifying a function that can be thought of as a running
estimate of the population of assorter.  One such function,
\texttt{shrinkTrunc()}, has two tuning parameters, $\eta_0$, which can be
thought of as an initial estimate of the true assorter mean for the ballots,
and $d$, which can be thought of as how much emphasis we put on $\eta_0$
(higher values) or how eagerly we learn from the sample (lower values).  In
this paper we explore the effect of those parameters on audit sample sizes.
Other parameters of ALPHA were set to the same values used
by~\cite{ek2023awaire}.


\section{Weighting Schemes}

In the AWAIRE framework, the assorter associated with each requirement $r$ from
some \altorder{}\footnote{%
The details in this section are analogous for every \altorder{}.}
is applied to the sample (at time $\ell$), forming the list of values
$(X^r_t)_{t=1}^\ell$.  Let $M_{r,\ell}$ be the test supermartingale for
requirement $r$ evaluated at time $\ell$, which can be written as a product of
increments:
\[
    M_{r,\ell} := \prod_{t=0}^{\ell} m_{r,t},
\]
where $m_{r,0} = 1$ denotes the starting value and $m_{r,t}$ reflects how
$X^r_t$ impacts the cumulative evidence that requirement $r$ is false.
For example $m_{r,t} < 1$ means that $X^r_t$ gives no evidence that $r$ is
false; $m_{r,t} > 1$ means that $X^r_t$ gives evidence that $r$ is false.
Because $M_{r,\ell}$ is a test supermartingale,
\begin{equation}
\label{eq:e-expect}
    \mathbb{E}(m_{r,t} \mid (X^r_\ell)_{\ell=0}^{t-1}) \leqslant 1,
\end{equation}
if $r$ is true.
These supermartingales are referred to as \emph{base} test supermartingales.

\subsection{Intersection Test Martingales}

Each \altorder{} has an \emph{intersection} test supermartingale, which
measures the cumulative evidence for that \altorder{} being the true
elimination order.  To correct for multiplicity, the intersection test
supermartingale uses a weighted average of its base test supermartingales.

Specifically,  let the weights at time $t$ be $\{w_{r,t}\}_{r=1}^{r_i}$.  These
can depend on the data collected up to time $t - 1$, but not on any later data.
Using those weights, the intersection test supermartingale is a product of
weighted combinations of the terms of the base test supermartingales:
\[
    M_\ell := \prod_{t=1}^\ell \frac{\sum_{r=1}^{r_i} w_{r,t} m_{r,t}}
                                    {\sum_{r=1}^{r_i} w_{r,t}},
    \quad \ell = 0, 1, \dots,
\]
with $M_0 := 1$.

The base test supermartingales for requirements that are false tend to grow
with $t$.  We explore how to make $M_\ell$ grow quickly by choosing efficient
weighting schemes.

\subsection{Previous Schemes}

The original AWAIRE paper~\cite{ek2023awaire} investigates a number of schemes
for weight selection:

\begin{description}
\item[Linear.]  Proportional to previous value, $w_{r,t} := M_{r,t-1}$.
\item[Quadratic.]  Proportional to the square of the previous value,
                                                $w_{r,t} := M_{r,t-1}^2$.
\item[Largest.]  Take only the largest base supermartingale(s) and ignore the
    rest, $w_{r,t} := 1$ if $r \in \argmax_{r'} M_{r',t-1}$; otherwise,
    $w_{r,t} := 0$.
\end{description}
Experiments in~\cite{ek2023awaire} found \textbf{Largest} to be the most robust
choice.  But there are many more weighting schemes possible, and indeed a
single weighting scheme may not be the best for different IRV elections.

\subsection{New Schemes: Variants of Previous Schemes}

We introduce several new weighting schemes usable within AWAIRE to try to
generate intersection test supermartingale that grow quickly.  The schemes in
the previous subsection are \emph{myopic}: they only look at the previous value
of the base supermartingales.  This makes them inefficient when two or more
base test supermartingales frequently swap leadership positions.  Below we
examine more complex weighting schemes, many of which look back at the test
supermartingale values of the last $i$ steps:
\begin{description}
\item[LargestCount($i$)] Put all weight on the base test supermartingale that
    was largest most often in the previous $i$ draws.  This is a less myopic
    version of \textbf{Largest}.
\item[LargestMean($i$)] Put all weight on the base test supermartingale whose
    mean in the last $i$ draws was largest.  Again this is a less myopic
    version of \textbf{Largest} that also takes into account the magnitude
    of the difference between different requirements.
\item[Linear$+$] Same as Linear, but if at least one base test supermartingale
    is greater than 1 at step $t-1$, put weight $0$ on all base test
    supermartingales that are less than 1.  This attempts to remove from
    consideration requirements that appear to be compatible with the data.
\item[LinearCount($i$)] Put weight in linear proportion to how many times each
    requirement has been the largest looking back $i$ steps.  This is fairer
    version of \textbf{LargestCount} that spreads its bets on requirements
    that have been largest.
\item[LinearMean($i$)] Taking the moving average value of each requirement
    looking back $i$ steps, put weight in linear proportion to their means.
    This is a less myopic version of \textbf{Linear}.
\item[Quadratic$+$] Same as Quadratic, adapted in an analogous way to
    Linear$+$.
\end{description}

\subsection{New Schemes: Portfolio Algorithms}
\label{sec:portfolio-algorithms}

There is a large literature on \emph{portfolio algorithms}, which aim to
maximise the growth of wealth in a stock market by selecting a portfolio of
stocks.  This involves selecting how to split some starting capital into
amounts to invest in each stock and how to re-invest the capital each period.
Our weighting schemes fit this paradigm, with the base test supermartingales
representing individual stocks and the weights corresponding to the fraction of
the current fortune invested in each stock in each time period.  Any portfolio
algorithm that only uses information about previous stock prices yields a
weighting scheme that could be used with AWAIRE.

We attempted to test a variety of portfolio algorithms, but the vast majority
of papers describing such algorithms do not include software.  The most
comprehensive collection of software we found was at:
\begin{center}
    \url{https://github.com/Marigold/universal-portfolios}
\end{center}
We tried to use these, but the only scheme that ran successfully was:
\begin{description}
\item[ONS($\delta$)] ``Online Newton Step'' with tuning parameter
    $\delta$~\cite{agarwal2006}.  This is a family of weighting schemes coming
    from investment portfolio management.
\end{description}
The other algorithms either did not apply to our problem or crashed due to
floating point overflow.  Even ONS($\delta$) sometimes crashed for
$\delta = 0.66, 1$, and sometimes $2$.  Thus, we analyse ONS with $\delta > 2$.

Previous work has shown that (under suitable conditions) the optimal portfolio
is a \emph{constant rebalanced portfolio} \cite{algoet1988,breiman1961optimal},
where at each timestep the fraction of the current capital invested in a given
stock is constant over time.  The optimal allocation, however, can only be
determined in hindsight.

A class of portfolio algorithms that are asymptotically optimal are
\emph{$F$-weighted portfolios}, also known as \emph{universal portfolios}
\cite{cover1991}.  However, they often perform poorly for small sample sizes
and do not necessarily have computationally efficient implementations
\cite{kalai2002efficient}.  Nevertheless, they might inspire better weighting
schemes; we discuss some ideas in \autoref{sec:ideas-for-better-schemes}.

\subsection{Software}

Our software implementation of the new weighting schemes, along with the AWAIRE
implementation, is available at:
\url{https://github.com/aekh/awaire}


\section{Analyses and Results}

\subsection{Data}

We used the same NSW 2015 Legislative Assembly election data as in the AWAIRE
paper \cite{ek2023awaire}, consisting of 71~contests with 6 or fewer
candidates.\footnote{%
\url{https://github.com/michelleblom/margin-irv}}
Experiments showed that the relative performance of the weighting schemes and
various tuning parameters for ALPHA depend on the margin of victory.  To
understand these differences more clearly, we partitioned the dataset into four
categories based on the margin of victory:
\begin{description}
\item[Huge.] Margins of 10\% and above. (41 contests)
\item[Large.] Margins in the range 4--10\%. (19 contests)
\item[Medium.] Margins in the range 1.5--4\%. (7 contests)
\item[Small.] Margins less than 1.5\%. (4 contests)
\end{description}

\subsection{Initial Comparison of Weighting Schemes}

First, let's compare the weighting schemes we have listed above.
We will use the $d=50$ and $\eta_0=0.52$ as was used in previous AWAIRE paper.
We will refer to this as \emph{the previous default}.
We used a risk limit of $5\%$.
See \autoref{fig:scheme-compare} for results.
For each weighting scheme, we ran 500 simulated audits for each contest.
First we calculated statistics across all simulated audits in each of the four
margin categories.

\begin{figure}[ptb]
\centering\includegraphics[width=0.99\textwidth]{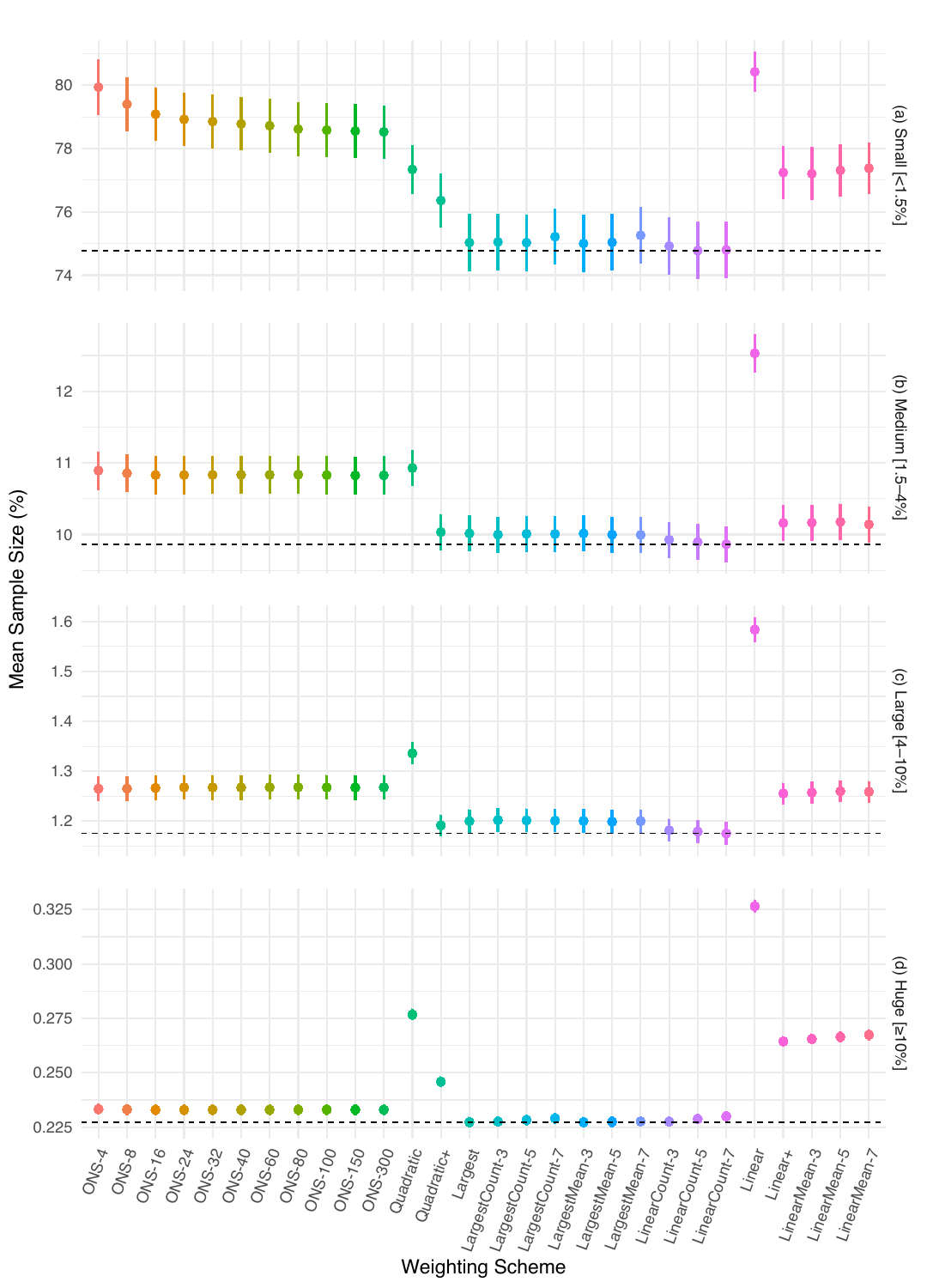}
\caption{Mean sample size (as a percentage of the total ballots in each
    contest; $\pm 2$ standard errors) across all simulated audits in each of
    the margin categories (rows). The vertical gridlines in panels (a)--(d)
    correspond respectively to approximately 500, 150, 25 and 10 ballots.  The
    dashed lines show the best mean sample size achieved within each panel.}
\label{fig:scheme-compare}
\end{figure}

\autoref{fig:scheme-compare} indicates that Quadratic$+$, LargestCount(5),
LinearCount(7), and the previously introduced Largest are consistently best
while also performing somewhat differently across the four categories.  The
following sections concentrate on those weighting schemes.

\subsection{Tuning Parameters for \texttt{shrinkTrunc()} in ALPHA}
\label{sec:tuning-parameters}

\begin{figure}[ptb]
\centering\includegraphics[width=\textwidth]{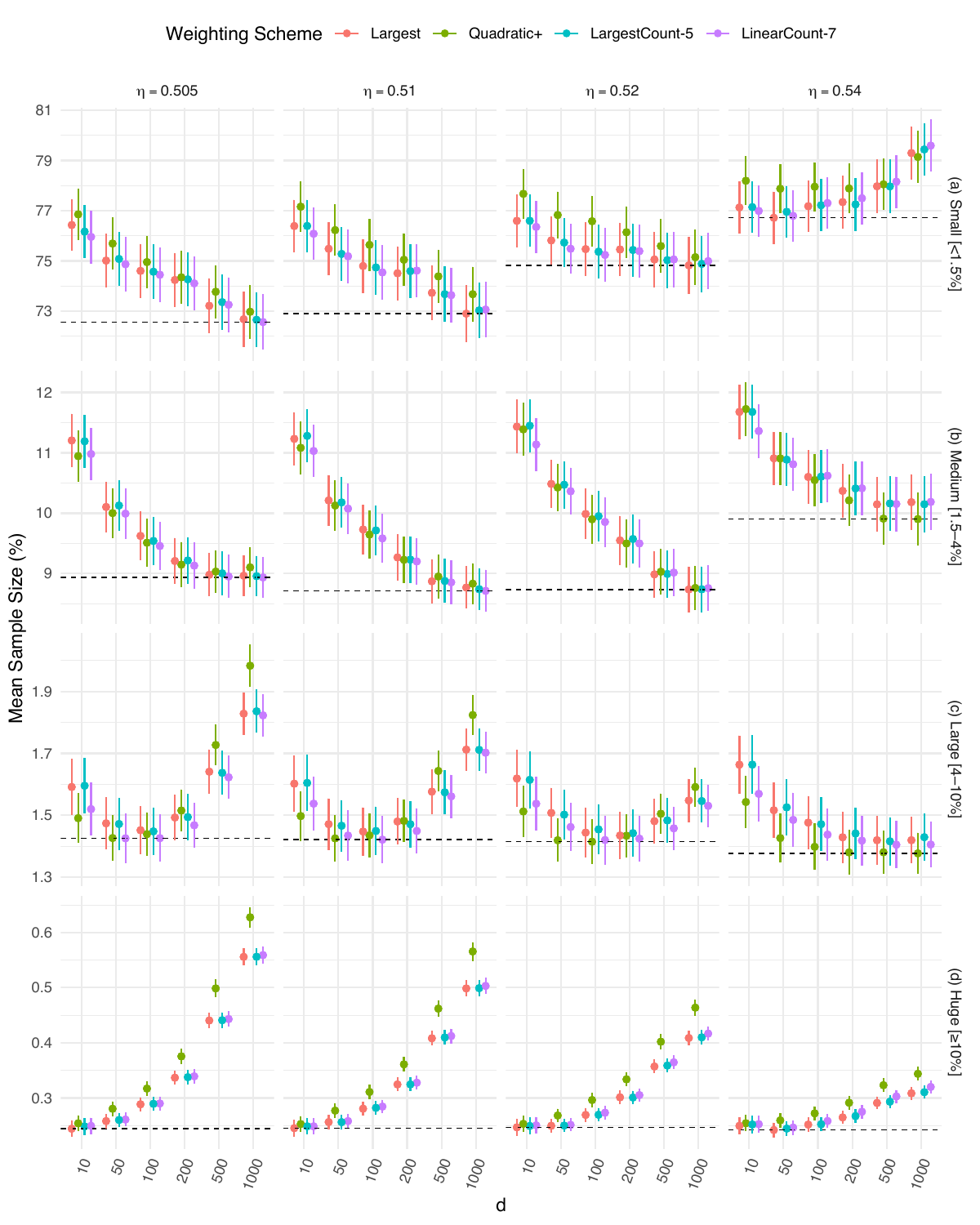}
\caption{Mean sample size (as a percentage of the total ballots in each
    contest; $\pm 2$ standard errors) across all simulated audits in each of
    the margin categories (rows) and settings for \texttt{shrinkTrunc()}
    ($\eta_0$ across columns and $d$ on the x-axis).  Three contests were
    selected to represent each category, see \autoref{sec:tuning-parameters}
    for details.  The dashed lines show the best mean sample size achieved
    within each panel.}
\label{fig:alpha-compare}
\end{figure}

The purpose of these experiments is twofold: first, to understand what the best
tuning parameters are for ALPHA when dealing with IRV contests; second, to
ensure that the evaluation of the weighting schemes is somewhat decoupled from
the choice of underlying test supermartingale.

We used $\eta_0 \in \{0.505, 0.51, 0.52, 0.54\}$ and $d \in \{10, 50, 100, 200,
500, 1000\}$.  For time reasons, for these experiments we used a subset of the
contests consisting of 3 elections per margin category: (a)~the contest with
the smallest margin, (b)~the contest with the largest margin, and (c)~a contest
in the middle (rounded to smaller margin if no true middle).

\autoref{fig:alpha-compare} shows the results from these experiments.  There
was no single best choice of tuning parameters, but $\eta_0 = 0.51$ and $d =
200$ were reasonable defaults.  Selecting $\eta_0$ and $d$ involves trade-offs
in performance across contests.  For example, with $\eta_0 = 0.51$, increasing
$d$ improved efficiency for Small-margin elections but decreased efficiency for
the Huge-margin category.  The default choice balances efficiency across the
categories by slightly prioritising good performance for Large and Medium at
the expense of Small and Huge.  Our reasoning is as follows:
\begin{itemize}
\item Audits for Huge-margin contests will generally only need small sample
    sizes, thus increasing the number of samples by a relatively large
    percentage has low absolute cost.
\item Audits of contests with very small margins may require sampling fractions
    so large that a full hand count is more efficient.
\end{itemize}

\subsection{Detailed Comparison of Selected Weighting Schemes}

From the results in \autoref{fig:alpha-compare}, we see two types of patterns:
either the difference between the weighting schemes is barely discernable, or
Quadratic$+$ differs clearly from the others (performing either better or
worse).  Since the other three methods performed so similarly, we recommend
using Largest because of its simplicity (it only requires storing values from 1
draw in the past).  Therefore, we selected only Quadratic$+$ and Largest for
further comparisons.

In this section, we will compare their performance with $\eta_0 = 0.51$ and
$d = 200$ (as selected in \autoref{sec:tuning-parameters}) against Largest with
$\eta_0 = 0.52$ and $d = 50$ (the default in \cite{ek2023awaire}).  We used all
contests with 6 or fewer candidates for the comparison.

\autoref{fig:new-vs-old} shows the average reduction in the mean sample size,
plotted against the margin of each contest.  The Largest and Quadratic$+$
schemes both perform similarly.  There is a substantial reduction in sampling
effort for elections with small-to-medium margins, and a very slight increase
for large-to-huge margins.  \autoref{fig:largest-vs-quadraticpositive} compares
the average reduction in mean sample size for the two new default choices in
more detail.  For the majority of contests, the Largest scheme is slightly
better than Quadratic$+$.

\begin{figure}[ptb]
\centering\includegraphics[width=\textwidth]{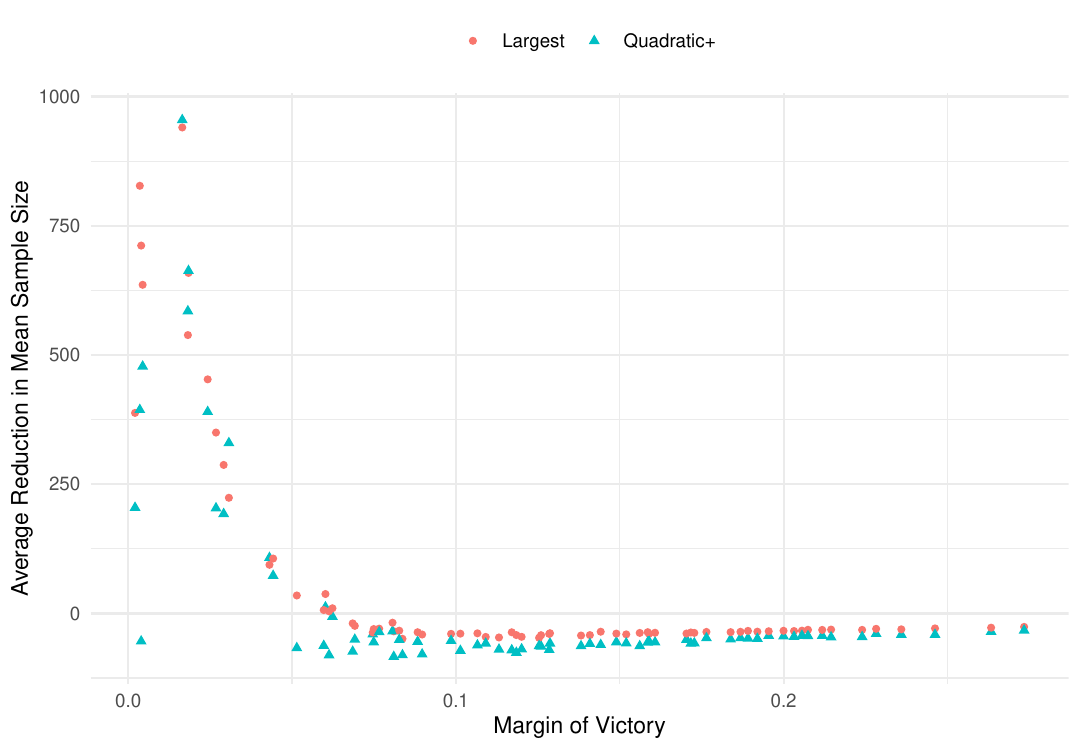}
\caption{Average reduction in mean sample size for two default choices compared
    to the previous default choice. Each point represents a single contest
    (averaged over 500 simulated audits). The margin (x-axis) is shown as a
    proportion out the total ballots in each contest.}
\label{fig:new-vs-old}
\end{figure}

\begin{figure}[ht]
\centering\includegraphics[width=0.8\textwidth]{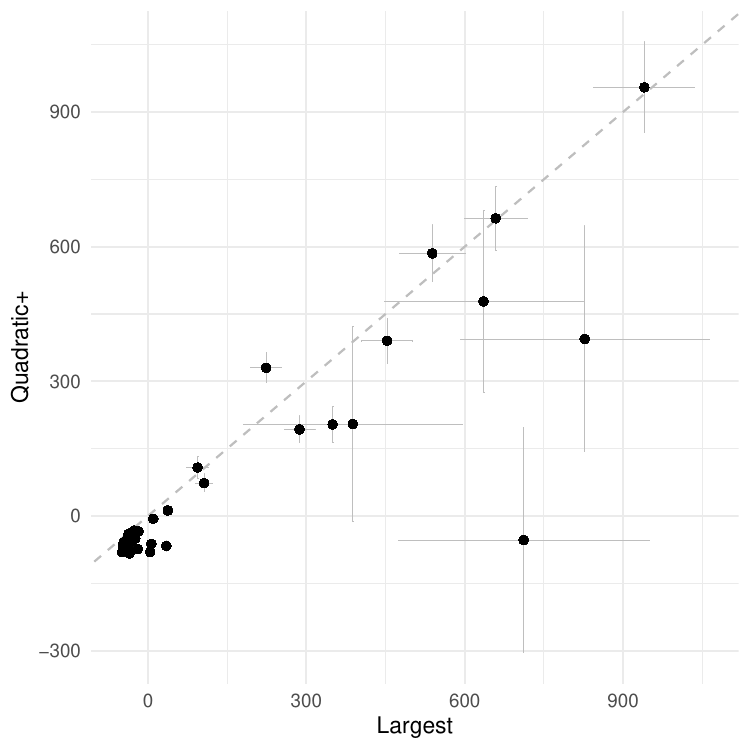}
\caption{Average reduction in mean sample size for our two default choices; now
    $\pm 1$ standard error in both directions.  Each point represents a single
    contest (averaged over 500 simulated audits). The majority of points are on
    the right-hand side of the diagonal, indicating a larger average reduction
    when using Largest as compared to Quadratic$+$.}
\label{fig:largest-vs-quadraticpositive}
\end{figure}


\section{Improving Weighting Schemes Using More Sophisticated Portfolio
Approaches}
\label{sec:ideas-for-better-schemes}

As discussed earlier in \autoref{sec:portfolio-algorithms}, previous
theoretical work has shown that the optimal weighting scheme will be a constant
rebalanced portfolio, for a set of weights that can only be determined in
hindsight.

We conjecture that for typical elections, the optimal set of weights
concentrates on a single requirement, and perhaps sometimes across a small
number of requirements (when some base test supermartingales frequently swap
leadership positions).

It would be interesting to explore this conjecture by approximating the optimal
constant rebalanced weighting scheme using some kind of optimisation algorithm
with the full set of ballots.  If the conjecture is true, then it would explain
why Largest and similar schemes often performed well in our comparisons.  In
elections where the conjecture is false, it would be worth exploring some more
complex schemes.

The class of $F$-weighted portfolio algorithms is natural to consider based on
asymptotic results, although we note that their short-run performance and
computational complexity are typically poor.

The Linear scheme is in fact an $F$-weighted portfolio algorithm, for a rather
restrictive choice of the distribution $F$; see
\autoref{thm:linear-is-F-weighted} below.  Most of the other schemes, including
Largest, are not in that class because they can change zero weights to non-zero
weights over time (not possible with an $F$-weighted algorithm).  However,
these schemes might be able to be approximated by an $F$-weighted algorithm, by
``rounding off'' weights that are very close to 0 or 1.

This suggests that we could work with more complex $F$-weighted portfolio
algorithms if we approximate them appropriately.  For example, consider the
``general universal portfolio'' by Cover \cite{cover1991}, which creates an $F$
that places positive mass on every face of a simplex.  We could mimic this in a
more heuristic and computationally efficient way by greedily grouping only the
best base martingales together and optimising the weight amongst them with a
general $F$.  Such a calculation would require applying possible weights
(within the group) across the full history every time the weights are updated,
which is more demanding than our current schemes, but it might be feasible for
a small group of requirements.

\begin{theorem}
\label{thm:linear-is-F-weighted}
Linear is an $F$-weighted portfolio algorithm.
\end{theorem}
\begin{proof}
Consider a set of requirements $R_i$.  Let $\vec{b} := (b_1, b_2, \dots,
b_{r_i})$ be a vector of nonnegative weights for the base test
supermartingales for the $r_i$ requirements in this set.  Let
$M_t(\vec{b})$ be the intersection test supermartingale obtained using the
weight vector $\vec{b}$ at each time step (a constant rebalanced portfolio).

An \emph{$F$-weighted portfolio} updates the weights at each time step using a
performance-weighted average of constant rebalanced portfolios and an initial
distribution $F$ across possible weight vectors:
\[
  \vec{b}_t = \frac{\int \vec{b} M_{t-1}(\vec{b}) F(\vec{b}) \, d\vec{b}}
                   {\int         M_{t-1}(\vec{b}) F(\vec{b}) \, d\vec{b}} .
\]
Let $\vec{b}_r = (0, 0, \dots, 1, \dots, 0)$, consisting of a weight 1 for the
$r$th component and 0 for all others.  Using these weights yields the base
test supermartingale for requirement~$r$.  In other words, $M_t(\vec{b}_r)
= M_{r, t}$.

Consider a distribution $F$ that places mass on all vectors $\vec{b}_r$, and
zero probability elsewhere: $F(\vec{b}) =  1/r_i \sum_{r=1}^{r_i}
\delta_{\vec{b}_r}(\vec{b})$, where $\delta$ is the Dirac delta function.  We
show that this produces the Linear weighting scheme:
\[
  \vec{b}_t = \frac{\int \vec{b} M_{t-1}(\vec{b}) 1/r_i \sum_{r=1}^{r_i} \delta_{\vec{b}_r}(\vec{b}) \, d\vec{b}}
                   {\int         M_{t-1}(\vec{b}) 1/r_i \sum_{r=1}^{r_i} \delta_{\vec{b}_r}(\vec{b}) \, d\vec{b}}
            = \frac{\sum_{r=1}^{r_i} \vec{b}_r M_{t-1}(\vec{b}_r)}
                   {\sum_{r=1}^{r_i}           M_{t-1}(\vec{b}_r)}
            = \frac{\sum_{r=1}^{r_i} \vec{b}_r M_{r, t-1}}
                   {\sum_{r=1}^{r_i}           M_{r, t-1}} .
\]
This is precisely the weight vector for the Linear scheme
($w_{r,t} := M_{r,t-1}$).
\qed
\end{proof}


\section{Discussion}

AWAIRE has many adjustable parameters including tuning parameters in the base
ALPHA test supermartingales and the adaptive selection of weights in combining
the base test supermartingales.  We explored an extensive range of weighting
schemes and of tuning parameters for \texttt{shrinkTrunc()} in ALPHA, providing
a deeper understanding of the trade-offs.  We provided recommendations for
default choices of the parameters in \texttt{shrinkTrunc()} for ALPHA and the
adaptive weights.

This work focused on auditing IRV contests when the election reports a winner
but does not report the interpretation of individual ballot cards (CVRs).
\cite{ek2023awaire} shows that reliable CVRs, if they are available, can make
AWAIRE more efficient.   In some jurisdictions, CVRs are not available but some
information about the election count is, such as round-by-round vote tallies.
It might be possible to use such tallies to tune AWAIRE parameters.  For
example, the last-round margin is often the margin of the contest as a whole,
or at least provides an upper bound.  This could be used to set $\eta_0$ to a
useful default value specific for that contest, rather than simply using our
default choice.

For any specific \altorder{}, the requirements will have a range of assorter
margins, each with a different optimal tuning for ALPHA.  Absent any
information (such as CVRs) to tune the tests individually, we proposed a
default value of $\eta_0$ to use for all requirements.  Large values of $d$
will make ALPHA adapt very slowly to the data, which will be helpful for some
requirements but reduce efficiency for others, as illustrated in
\autoref{fig:new-vs-old}.

Our work has useful implications for a ``lazy'' implementation of AWAIRE that
decreases the computational burden.  Essentially, only schemes that have sparse
weights (such as Largest) are feasible.  The fact that Largest and its variants
were among the best schemes suggests that a lazy implementation should not
incur a large penalty in statistical performance.  The major challenge will be
to ensure that a good requirement is found early on in any lazy algorithm, but
once that is done the audit should perform competitively without needing to
explore for more requirements.

It was difficult to find many practically useful software implementations of
existing portfolio algorithms.  That limited how many we could include in our
comparison.  However, many portfolio algorithms are known to be either
computationally inefficient, or only asymptotically optimal but perform poorly
for small time horizons; thus, they would not fare well in our comparisons
anyway.  It may be the case that some existing algorithms would perform better
than all of the ones we have tried thus far.  Future work can explore
implementing any algorithms that are expected to be computationally efficient
and also perform well on short time horizons.

There may be a theoretically best weighting scheme that could be determined
from the complete set of ballots (i.e., ``in hindsight'').  Future work could
investigate optimal weighting and ways to approximate it efficiently and
adaptively in practice.

It would be interesting to see to what extent the theoretically best schemes
place nearly all of their weight on only a few requirements.  We suspect this
might be the case, given how well the Largest scheme performs in our
comparisons.


\subsubsection{Acknowledgements.}

We thank Michelle Blom and Vanessa Teague for helpful discussions and
suggestions.
This work was supported by the Australian Research Council
(Discovery Project DP220101012, OPTIMA ITTC IC200100009)
and the U.S.\ National Science Foundation (SaTC~2228884).


\bibliographystyle{splncs04}
\bibliography{references}


\end{document}